\documentclass[graybox]{svmult}

\usepackage{type1cm}        
\usepackage{makeidx}         
\usepackage{graphicx}        
\usepackage{multicol}        
\usepackage[bottom]{footmisc}

\usepackage{newtxtext}       %

\usepackage{enumerate}
\usepackage{cite}
\usepackage{enumerate}
\usepackage{cite}
\usepackage{textcomp}
\usepackage{amsmath, amssymb, amsfonts}
\usepackage{mathrsfs}
\usepackage{mathtools}
\usepackage{dsfont}
\usepackage{accents}
\usepackage{color}
\usepackage{caption}

\DeclareMathOperator{\tr}{tr}
\DeclareMathOperator{\minimize}{minimize}

\DeclareMathOperator{\voi}{VoI}
\DeclareMathOperator{\EXP}{\mathsf{E}}
\DeclareMathOperator{\Cov}{\mathsf{cov}}

\DeclareMathOperator{\Prob}{\mathsf{p}}

\DeclareMathOperator{\aoi}{AoI}

\begin{document}

\renewcommand\footnotemark{}
\title*{Relation between Value and Age of Information in Feedback Control \thanks{\hspace{-4.2mm}Corresponding Author: Touraj Soleymani (touraj@kth.se). Published as a book chapter by \emph{Cambridge University Press}.}}

\titlerunning{Relation between Value and Age of Information in Feedback Control}
\author{Touraj Soleymani, John S. Baras, and Karl H. Johansson}

\institute{Touraj Soleymani \at Royal Institute of Technology, Sweden
\and John S. Baras \at University of Maryland, United States
\and Karl H. Johansson \at Royal Institute of Technology, Sweden}

\maketitle

\vspace{-10mm}
\abstract{In this chapter, we investigate the value of information as a more comprehensive instrument than the age of information for optimally shaping the information flow in a networked control system. In particular, we quantify the value of information based on the variation in a value function, and discuss the structural properties of this metric. Through our analysis, we establish the mathematical relation between the value of information and the age of information. We prove that the value of information is in general a function of an estimation discrepancy that depends on the age of information and the primitive variables. In addition, we prove that there exists a condition under which the value of information becomes completely expressible in terms of the age of information. Nonetheless, we show that this condition is not achievable without a degradation in the performance of the system.}

\section{Introduction}
This chapter is concerned with networked control systems\index{system! networked control}, which are distributed feedback systems\index{system! distributed feedback} where the underlying components, i.e., sensors\index{sensor}, actuators\index{actuator}, and controllers\index{controller}, are connected to each other via communication channels\index{channel}. Such systems have flexible architectures with potential applications in a wide range of areas such as autonomous driving, remote surgery, and space exploration. In these systems, the information that flows through the communication channels plays a key role. Consider, for instance, an unmanned vehicle that should be navigated remotely. For the purpose of the navigation of this vehicle, the sensory information needs to be transmitted from the sensors of the vehicle to a remote controller, and in turn the control commands need to be transmitted from the controller to the actuators of the vehicle. Note that, on the one hand, there exist often different costs and constraints that can hinder the timely updates of the components of a networked control system. On the other hand, it is not hard to see that steering a networked control system based on quite out-of-date information could result in a catastrophic failure. Given this situation, our objective here is to put forward a systematic way for optimally shaping the information flow in a networked control system such that a specific level of performance is attained.

As discussed in the various chapters of the present book, the \emph{age of information}~\cite{yates2021age}, which measures the freshness of information at each time, is an appropriate instrument for shaping the information flow in many networked real-time systems\index{system! networked real-time}. However, we show in this chapter that a more comprehensive instrument than the age of information is required when one deals with networked control systems. This other metric is the \emph{value of information}~\cite{touraj-thesis, voi, voi2}, which measures the difference between the benefit and the cost of information at each time.

In our study, we consider a basic networked control system where the channel between the sensor and the controller is costly, and the channel between the controller and the actuator is cost-free. We begin our analysis by making a trade-off that is defined between the packet rate and the regulation cost. The decision makers are an event trigger and the controller. The event trigger decides about the transmission of information from the sensor to the controller at each time, and the controller decides about the control input for the actuator at each time. For the purpose of this study, we design the controller based on the certainty-equivalence principle, and mainly focus on the design of the event trigger. We show that the optimal triggering policy permits a transmission only when the value of information is nonnegative. We quantify the value of information at each time as the variation in a value function with respect to a piece of information that can be communicated to the controller. Through our analysis, we establish the mathematical relation between the value of information and the age of information. We prove that the value of information is in general a function of an estimation discrepancy that depends on the age of information and the primitive variables. In addition, we prove that there exists a condition associated with the information structure of the system under which the value of information becomes completely expressible in terms of the age of information. Nonetheless, we show that this condition is not achievable without a degradation in the performance of the system.

There are multiple works that are closely related to our study~\cite{lipsa2011, molin2017, chakravorty2016, rabi2012, guo2021-IT, guo2021-TAC, sun2019}. In these works, optimal triggering policies were characterized in settings that are different from ours. In particular, Lipsa and Martins~\cite{lipsa2011} used majorization theory to study the estimation of a scalar Gauss-Markov process, and proved that the optimal triggering policy is symmetric threshold. Molin and Hirche~\cite{molin2017} studied the convergence properties of an iterative algorithm for the estimation of a scalar Markov process with arbitrary noise distribution, and found a result coinciding with that in~\cite{lipsa2011}. Chakravorty and Mahajan~\cite{chakravorty2016} investigated the estimation of a scalar autoregressive Markov process with symmetric noise distribution based on renewal theory, and proved that the optimal triggering policy is symmetric threshold. Rabi~\emph{et~al.}\cite{rabi2012} formulated the estimation of the scalar Wiener and scalar Ornstein-Uhlenbeck processes as an optimal multiple stopping time problem by discarding the signaling effect, and showed that the optimal triggering policy is symmetric threshold. Guo and Kostina \cite{guo2021-IT} addressed the estimation of the scalar Wiener and scalar Ornstein-Uhlenbeck processes in the presence of the signaling effect, and obtained a result that matches with that in~\cite{rabi2012}. They also looked at the estimation of the scalar Wiener process with fixed communication delay in the presence of the signaling effect in~\cite{guo2021-TAC}, and obtained a similar structural result. Furthermore, Sun~\emph{et~al.}~\cite{sun2019} studied the estimation of the scalar Wiener process with random communication delay by discarding the signaling effect, and showed that the optimal triggering policy is symmetric threshold. In contrast to the above works, we here focus on the estimation of a Gauss-Markov process\index{process! Gauss-Markov} with random processing delay, and characterize the optimal triggering policy. Our study builds on the framework we developed previously for the value of information in~\cite{touraj-thesis, voi, voi2, mywodespaper}.

This chapter is organized in the following way. We formulate the problem in Section~\ref{sec2}. Then, we present the main results in Section~\ref{sec3}, and provide a numerical example in Section~\ref{sec4}. Finally, we conclude the chapter in Section~\ref{sec5}.

\section{Problem Statement}\label{sec2}
In this section, we mathematically formulate the rate-regulation trade-off as a stochastic optimization problem. In our setting, the process under control satisfies the following equations with an $n$-dimensional state, an $m$-dimensional input, and an $n$-dimensional output:
\begin{align}
	x_{k+1} &= A x_k + B u_k + w_k \label{c1:eq:sys}\\[2\jot]
	y_k &= x_{k-\tau_k} \label{c1:eq:output}	
\end{align}
for $ k \in \mathcal{K} = \{0,1,\dots,N\}$ with initial condition $x_0$, where $x_k \in \mathbb{R}^n$ is the state of the process, $A \in \mathbb{R}^{n \times n}$ is the state matrix, $B \in \mathbb{R}^{n \times m}$ is the input matrix, $u_k \in \mathbb{R}^m$ is the control input decided by the controller\index{controller}, $w_k \in \mathbb{R}^n$ is a Gaussian white noise with zero mean and covariance $W \succ 0$, $y_k \in \mathbb{R}^n$ is the output of the process, $\tau_k \in \mathbb{N}_0$ is a random processing delay\index{delay! random processing} with known probability distribution, and $N$ is a finite time horizon. It is assumed that $x_0$ is a Gaussian vector with mean $m_0$ and covariance $M_0$, that $\tau_0 = 0$, and that $x_0$, $w_k$, and $\tau_k$  are mutually independent for all $k \in \mathcal{K}$.

The random processing delay can be due to various sensing constraints. One example of the output model is:
\begin{align}\label{outputmodelused}
y_k = \left\{
  \begin{array}{l l}
     x_k, & \ \text{with probability}\ p_k, \\
     x_{k-d}, & \ \text{otherwise}
  \end{array} \right.
\end{align}
where $d$ is a fixed delay and $p_k$ is a probability. In (\ref{outputmodelused}), the output of the process at time $k$ is either the current state or a delayed state. Note that when no sensing constraints exist, we have $y_k = x_k$ for $k \in \mathcal{K}$.

The channel that connects the sensor to the controller is costly, errorless, and with one-step delay. Let $\delta_{k} \in \{0,1\}$ be the transmission variable in this channel decided the event trigger\index{event trigger}. Then, $y_k$ is transmitted over the channel and received by the controller after one-step delay if $\delta_k = 1$. Otherwise, nothing is transmitted and received. It is assumed that, in a successful transmission, the total delay of the received observation can be detected by the controller. For the purpose of this study, we design the controller based on the certainty-equivalence principle without the signaling effect (see \cite{voi, voi2} for more details). This allows us to concentrate on the design of the event trigger. Let the information sets of the event trigger and the controller, including only causal information, at time $k$ be denoted by $\mathcal{I}^e_k$ and $\mathcal{I}^c_k$, respectively. We say that a triggering policy $\pi$ is admissible if $\pi = \{\Prob(\delta_k | \mathcal{I}^e_k) \}_{k=0}^{N}$, where $\Prob(\delta_k | \mathcal{I}^e_k)$ is a Borel measurable transition kernel. We represent the admissible triggering policy set by $\mathcal{P}$. The rate-regulation trade-off\index{trade-off! rate-regulation} between the packet rate\index{packet rate} and the regulation cost\index{regulation cost} can then be expressed by the following stochastic optimization problem\index{problem! stochastic optimization}:
\begin{align}\label{eq:main_problem1}
	\underset{\pi \in \mathcal{P}}{\minimize} \ \Phi(\pi)
\end{align}
where
\begin{align}
	\Phi(\pi) =  \textstyle  \EXP\Big[\frac{(1-\lambda)}{N+1}\sum_{k=0}^{N} \ell \delta_k  +  \textstyle \frac{\lambda}{N+1} \textstyle\sum_{k=0}^{N} \big(\| x_{k+1}\|^2_{Q} + \|u_k\|^2_{R} \big) \Big]
\end{align}
where $\lambda \in (0,1)$ is the trade-off multiplier, $\ell$ is a weighting coefficient, $Q \succeq 0$ and $R \succ 0$ are weighting matrices, and $\|.\|$ represents the Euclidean norm. In the following, we seek the optimal triggering policy $\pi^\star$ associated with the problem in~(\ref{eq:main_problem1}).

\section{Main Results}\label{sec3}
First, we provide a formal definition of the age of information\index{age of information}, a metric that measures the freshness of information at a component at each time.

\begin{definition}\label{def:aoi}
The age of information at a component at time $k$ is the time elapsed since the generation of the freshest observation received by that component, i.e.,
\begin{align}
\aoi_k = k - t
\end{align}
where $t \leq k$ is the time of the freshest received observation.
\end{definition}

Let $\zeta_k \in [0, \zeta_{k-1} + 1]$ be the age of information at the event trigger. This implies that 
\begin{align}
\zeta_{k} = \left\{
  \begin{array}{l l}
     \tau_k, & \ \text{if} \ \tau_k < \zeta_{k-1}+1, \\
     \zeta_{k-1} + 1, & \ \text{otherwise}
  \end{array} \right.
\end{align}
for $k \in \mathcal{K}$ with initial condition $\zeta_0 = 0$, given the assumption $\tau_0=0$. Moreover, let $\eta_k \in [0, \eta_{k-1} + 1]$ be age of information at the controller. This implies that
\begin{align}\label{age-cont}
\eta_{k} = \left\{
  \begin{array}{l l}
     \zeta_{k-1} +1, & \ \text{if} \ \delta_{k-1} = 1, \\
     \eta_{k-1} + 1, & \ \text{otherwise}
  \end{array} \right.
\end{align}
for $k \in \mathcal{K}$ with initial condition $\eta_0 = \infty$, by convention. We say that the generated observation $x_{k-\tau_k}$ at time $k$ is informative if $\tau_k < \zeta_{k-1}+1$. Otherwise, we say it is obsolete. Since the obsolete observations\index{observation! obsolete} can safely be discarded, the information set of the event trigger can be determined by the set of the informative observations\index{observation! informative}, the communicated informative observations, and the previous decisions, i.e., $\mathcal{I}^{e}_k = \{ x_{t-\zeta_t}, x_{t-\eta_t}, \delta_s, u_s \ | \ 0 \leq t \leq k, 0 \leq s < k \}$, and the information set of the controller by the set of the communicated informative observations and the previous decisions, i.e., $\mathcal{I}^{c}_k = \{ x_{t-\eta_t}, \delta_s, u_s \ | \ 0 \leq t \leq k, 0 \leq s < k  \}$.

Given these information sets, in the next two lemmas, we derive the minimum mean-square-error estimators\index{estimator! minimum mean-square-error} at the event trigger and at the controller.

\begin{lemma}\label{prop:1lemma}
The conditional mean $\EXP[x_k | \mathcal{I}^e_k]$ is the minimum mean-square-error estimator at the event trigger, and satisfies
\begin{align}\label{eq:est-at-et}
	\EXP[x_k | \mathcal{I}^e_k] = A^{\zeta_k} x_{k-\zeta_k} + \textstyle \sum_{t=1}^{\zeta_k} A^{t-1} B u_{k-t}
\end{align}
for $k \in \mathcal{K}$.
\end{lemma}

\begin{proof}
Given $\mathcal{I}^e_k$, it is easy to see that $\EXP[x_k | \mathcal{I}^e_k]$ minimizes the mean-square error at the event trigger. Moreover, from the definition of $\zeta_k$, $x_{k-\zeta_k}$ represents the freshest observation at the event trigger. Writing $x_k$ in terms of $x_{k-\zeta_k}$ and taking the conditional expectation with respect to $\mathcal{I}^e_k$, we obtain the result.
\end{proof}

\begin{lemma}\label{prop:2lemma}
The conditional mean $\EXP[x_k | \mathcal{I}^c_k]$ is the minimum mean-square-error estimator at the controller, and satisfies
\begin{align}\label{eq:est-at-cont}
\EXP[x_{k+1} | \mathcal{I}^c_{k+1}] = \left\{
  \begin{array}{l l}
     A^{\zeta_k+1} x_{k-\zeta_k} + \textstyle \sum_{t=0}^{\zeta_k} A^{t} B u_{k-t}, & \ \text{if} \ \delta_{k} = 1, \\[1.75\jot]
     A \EXP[x_{k} | \mathcal{I}^c_{k}] + B u_k + \imath_k, & \ \text{otherwise}
  \end{array} \right.
\end{align}
for $k \in \mathcal{K}$ with initial condition $\EXP[x_{0} | \mathcal{I}^c_{0}] = m_0$, where $\imath_k = A ( \EXP[ x_k | \mathcal{I}^c_k, \delta_k =0 ] - \EXP[ x_k | \mathcal{I}^c_k])$ is the signaling residual.
\end{lemma}

\begin{proof}
Given $\mathcal{I}^c_k$, it is easy to see that $\EXP[x_k | \mathcal{I}^c_k]$ minimizes the mean-square error at the controller. In addition, writing $x_{k+1}$ in terms of $x_{k - \zeta_k}$ and taking the conditional expectation with respect to $\mathcal{I}^c_{k+1}$, we can write
\begin{align}\label{eq:expecation-prop2}
	\EXP[x_{k+1} | \mathcal{I}^c_{k+1}] = A^{\zeta_k+1} \EXP[x_{k-\zeta_k}| \mathcal{I}^c_{k+1}] + \textstyle \sum_{t=0}^{\zeta_k} A^{t} B u_{k-t}
\end{align}
where we used the fact that $\EXP[w_{k-t} | \mathcal{I}^c_{k+1}] = 0$ for all $t \in [0,\zeta_k]$ since the freshest observation that the controller might receive at time $k+1$ is $x_{k-\zeta_k}$. Now, note that if $\delta_k = 1$, the controller receives $x_{k-\zeta_k}$ at time $k+1$. In this case, $\EXP[x_{k-\zeta_k}| \mathcal{I}^c_{k+1}] = x_{k-\zeta_k}$. Inserting this into (\ref{eq:expecation-prop2}), we obtain
\begin{align*}
	\EXP[x_{k+1} | \mathcal{I}^c_{k+1}] = A^{\zeta_k+1} x_{k-\zeta_k} + \textstyle \sum_{t=0}^{\zeta_k} A^{t} B u_{k-t}.
\end{align*}
However, if $\delta_k = 0$, the controller receives nothing at time $k+1$. In this case, $\EXP[ x_{k-\zeta_k} | \mathcal{I}^c_{k+1}] = \EXP[ x_{k-\zeta_k} | \mathcal{I}^c_{k}, \delta_k=0]$. Hence, we can write $\EXP[x_{k-\zeta_k} | \mathcal{I}_{k}^c, \delta_k =0] = \EXP[x_{k-\zeta_k}| \mathcal{I}^c_k] + \tilde{x}_k$ for an appropriate residual $\tilde{x}_k$. Inserting this into (\ref{eq:expecation-prop2}), we obtain
\begin{align*}
	\EXP[x_{k+1} | \mathcal{I}^c_{k+1}] &= A^{\zeta_k+1} \EXP[x_{k-\zeta_k}| \mathcal{I}^c_{k}] + \textstyle \sum_{t=0}^{\zeta_k} A^{t} B u_{k-t} +  A^{\zeta_k+1} \tilde{x}_k\\[3\jot]
	&= A \EXP[ x_k | \mathcal{I}^c_k] + B u_k + \imath_k
\end{align*}
where we used the definition of $\EXP[ x_k | \mathcal{I}^c_k]$, and introduced $\imath_k =  A^{\zeta_k+1} \tilde{x}_k$. Lastly, we can write $\imath_k$ as
\begin{align*}
	\imath_k &= A^{\zeta_k+1} \big( \EXP[x_{k-\zeta_k} | \mathcal{I}_{k}^c, \delta_k =0] - \EXP[x_{k-\zeta_k}| \mathcal{I}^c_k] \big)
	\\[2.75\jot]
	&= A \big( \EXP[ x_k | \mathcal{I}^c_k, \delta_k =0 ] - \EXP[ x_k | \mathcal{I}^c_k] \big).
\end{align*}
This completes the proof.
\end{proof}

Now, we can express the control inputs based on the certainty equivalence principle\index{principle! certainty equivalence} as discussed. This leads to the employment of the linear-quadratic regulator\index{regulator! linear-quadratic} with a state estimate at the controller without the signaling residual. More specifically, we have $u_k = - L_k \hat{x}_k$, where $L_k = (B^T S_{k+1} B + R)^{-1} B^T S_{k+1} A$ is the linear-quadratic regulator gain, $\hat{x}_k$ is a state estimate at the controller satisfying (\ref{eq:est-at-cont}) with $\imath_k = 0$ for all $k \in \mathcal{K}$, and $S_k \succeq 0$ is a matrix that satisfies the algebraic Riccati equation\index{equation! algebraic Riccati}
\begin{equation}\label{eq:riccati}
\begin{aligned}
S_k &= Q + A^T S_{k+1} A - A^T S_{k+1} B (B^T S_{k+1} B + R)^{-1} B^T S_{k+1} A
\end{aligned}
\end{equation} 
for $k \in \mathcal{K}$ with initial condition $S_{N+1} = Q$.

We continue our analysis by presenting four lemmas that provide some properties associated with the estimation error\index{estimation error} $e_k = x_k - \hat{x}_k$ and the estimation mismatch\index{estimation mismatch} $\tilde{e}_k = \check{x}_k - \hat{x}_k$, where $\check{x}_k$ is the minimum mean-square-error state estimate at the event trigger.

\begin{lemma}\label{lemma1}
The estimation error $e_k$ satisfies
\begin{align}\label{eq:est-error-at-cont}
e_{k+1} = (1-\delta_k) (A e_k + w_k) + \delta_k \textstyle \sum_{t=0}^{\zeta_k} A^t w_{k-t}
\end{align}
for $k \in \mathcal{K}$ with initial condition $e_0 = x_0 - m_0$.
\end{lemma}
\begin{proof}
To prove, we need to subtract (\ref{eq:est-at-cont}) from (\ref{c1:eq:sys}) given $\imath_k = 0$ for all $k \in \mathcal{K}$.
\end{proof}

\begin{lemma}\label{lemma2}
The following facts are true:
\begin{align}
	\EXP[ e_{k+1} | \mathcal{I}^e_k] &= (1-\delta_k) A \tilde{e}_k\\[1.75\jot]
	\Cov[ e_{k+1} | \mathcal{I}^e_k] &= \textstyle \sum_{t=0}^{\zeta_k} A^t W {A^t}^T.
\end{align}	
\end{lemma}

\begin{proof}
By Lemma~\ref{lemma1}, when $\delta_k = 1$, we obtain
\begin{align*}
	\EXP[ e_{k+1} | \mathcal{I}^e_k] &= 0\\[1.75\jot]
	\Cov[ e_{k+1} | \mathcal{I}^e_k] &= \textstyle \sum_{t=0}^{\zeta_k} A^t W {A^t}^T.
\end{align*}
However, when $\delta_k =0$, we obtain
\begin{align*}
	\EXP[ e_{k+1} | \mathcal{I}^e_k] &= A \EXP[ e_k | \mathcal{I}^e_k] = A \tilde{e}_k\\[1.75\jot]
	\Cov[ e_{k+1} | \mathcal{I}^e_k] &= A \Cov[e_k | \mathcal{I}^e_k] A^T + W = \textstyle \sum_{t=0}^{\zeta_k} A^t W {A^t}^T
\end{align*}
where we used the fact that $\Cov[e_k | \mathcal{I}^e_k] = \Cov[ x_k | \mathcal{I}^e_k]$.
\end{proof}
\vspace{-4mm}
\begin{lemma}\label{lemma3}
The estimation mismatch $\tilde{e}_k$ satisfies
\begin{align}
	\tilde{e}_{k+1} = (1-\delta_k) A \tilde{e}_k + \textstyle \sum_{t = \zeta_{k+1}}^{\zeta_k} A^{t} w_{k-t}
\end{align}
for $k \in \mathcal{K}$ with initial condition $\tilde{e}_0 = x_0 - m_0$.
\end{lemma}

\begin{proof}
By Lemmas~\ref{prop:1lemma} and \ref{prop:2lemma} and from the definitions of $\zeta_k$ and $\eta_k$, we can write
\begin{align}
	\check{x}_k = A^{\zeta_k} x_{k-\zeta_k} + \textstyle \sum_{t=1}^{\zeta_k} A^{t-1} B u_{k-t}\label{lemma5-1}\\[1.75\jot]
	\hat{x}_k = A^{\eta_k} x_{k-\eta_k} + \textstyle \sum_{t=1}^{\eta_k} A^{t-1} B u_{k-t}. \label{lemma5-3}
\end{align}
Since $\zeta_k \leq \eta_k$, we can write $x_{k-\zeta_k}$ in terms of $x_{k - \eta_k}$ as
\begin{align}\label{lemma5-2}
	x_{k-\zeta_k} = A^{\eta_k - \zeta_k} x_{k - \eta_k} + \textstyle \sum_{t = 1}^{\eta_k - \zeta_k} (A^{t-1} B u_{k-\zeta_k -t} + A^{t-1} w_{k - \zeta_k - t}).
\end{align}
Inserting (\ref{lemma5-2}) in (\ref{lemma5-1}), we get
\begin{align*}
	\check{x}_k &= A^{\eta_k} x_{k-\eta_k} + \textstyle \sum_{t=1}^{\zeta_k} A^{t-1} B u_{k-t}\\[1.75\jot]
	& \qquad \qquad \qquad + \textstyle \sum_{t = 1}^{\eta_k - \zeta_k} (A^{t + \zeta_k -1} B u_{k-\zeta_k -t} + A^{t + \zeta_k -1} w_{k - \zeta_k - t})\\[1.75\jot]
	&= A^{\eta_k} x_{k-\eta_k} + \textstyle \sum_{t=1}^{\eta_k} A^{t-1} B u_{k-t} + \textstyle \sum_{t = \zeta_k + 1}^{\eta_k} A^{t -1} w_{k - t}.
\end{align*}
Now, using (\ref{lemma5-3}), we find
\begin{align}
	\tilde{e}_k &= \textstyle \sum_{t = \zeta_k + 1}^{\eta_k} A^{t -1} w_{k - t}\label{lemma5-ek}\\[1.75\jot]
	\tilde{e}_{k+1} &= \textstyle \sum_{t=\zeta_{k+1} +1}^{\eta_{k+1}} A^{t-1} w_{k+1-t}.
\end{align}
From (\ref{age-cont}), we know that the dynamics of $\eta_{k+1}$ depends on $\delta_k$. In particular, when $\delta_k = 1$, we have $\eta_{k+1} = \zeta_{k} + 1$, and can write
\begin{align}\label{lemma5-delta1}
	\tilde{e}_{k+1} = A^{\zeta_{k+1}} w_{k-\zeta_{k+1}} + \dots + A^{\zeta_{k}} w_{k-\zeta_k}.
\end{align}	
However, when $\delta_k = 0$, we have $\eta_{k+1} = \eta_{k} + 1$, and can write
\begin{align}\label{lemma5-delta0}
	\tilde{e}_{k+1} = A^{\zeta_{k+1}} w_{k-\zeta_{k+1}} + \dots + A^{\eta_{k}} w_{k-\eta_k}.
\end{align}
Hence, putting together (\ref{lemma5-delta1}) and (\ref{lemma5-delta0}), we get
\begin{align*}
	\tilde{e}_{k+1} &= (1-\delta_k) (A^{\zeta_{k}+1} w_{k-\zeta_k-1} + \dots + A^{\eta_{k}} w_{k-\eta_k}) \\[1.75\jot]
	&\qquad \qquad \qquad + A^{\zeta_{k+1}} w_{k-\zeta_{k+1}} + \dots + A^{\zeta_{k}} w_{k-\zeta_k}\\[1.75\jot]
	&= (1-\delta_k) A \tilde{e}_k + A^{\zeta_{k+1}} w_{k-\zeta_{k+1}} + \dots + A^{\zeta_{k}} w_{k-\zeta_k}
\end{align*}
where in the second equality we used (\ref{lemma5-ek}) and the fact that $\zeta_k \leq \eta_k$. This completes the proof.
\end{proof}

\begin{lemma}\label{lemma4}
Let $f(\tilde{e}_{k+1}):\mathbb{R}^n \to \mathbb{R}$ be a symmetric function of $\tilde{e}_{k+1}$. Then, $g(\tilde{e}_k,\delta_k) = \EXP[f(\tilde{e}_{k+1}) | \tilde{e}_k, \delta_k]:\mathbb{R}^n \times \{0,1\} \to \mathbb{R}$ is a symmetric function of $\tilde{e}_k$.
\end{lemma}

\begin{proof}
By Lemma~\ref{lemma3}, we can write
\begin{align}\label{eq:e-tilde}
	\tilde{e}_{k+1} = (1-\delta_k) A \tilde{e}_k + n_k
\end{align}
where $n_k = \textstyle \sum_{t = \zeta_{k+1}}^{\zeta_k} A^{t} w_{k-t}$ is a Gaussian white noise with zero mean and covariance $N_k = \textstyle \sum_{t = \zeta_{k+1}}^{\zeta_k} A^{t} W {A^t}^T$. Define the variable $\bar{n}_k$ as $\bar{n}_k = - n_k$. Then, $\bar{n}_k$ is also a Gaussian variable with zero mean and covariance $N_k$. Therefore, we have
\begin{align*}
	g(\tilde{e}_k,\delta_k) &= \EXP \Big[ f( \tilde{e}_{k+1}) \big| \tilde{e}_k, \delta_k \Big]\\[1\jot]
	&= \EXP \Big[f \big( (1-\delta_k) A \tilde{e}_k + n_k \big) \big| \tilde{e}_k, \delta_k \Big]\\[1\jot]
	&= \EXP \Big[f \big(-(1-\delta_k) A \tilde{e}_k - n_k \big) \big| \tilde{e}_k, \delta_k \Big]\\[1\jot]
	&= \textstyle \int_{\mathbb{R}^n} f \big(-(1-\delta_k) A \tilde{e}_k - n_k \big) c \exp(-\frac{1}{2} n_k^T N_k^{-1} n_k) \ d n_k\\[1.75\jot]
	&= \textstyle \int_{\mathbb{R}^n} f \big(-(1-\delta_k) A \tilde{e}_k + \bar{n}_k \big) c \exp(-\frac{1}{2} \bar{n}_k^T N_k^{-1} \bar{n}_k) \ d \bar{n}_k\\[1\jot]
	&= \EXP \Big[f \big(-(1-\delta_k) A \tilde{e}_k + n_k \big) \big| \tilde{e}_k, \delta_k \Big]\\[1\jot]
	&= g(-\tilde{e}_k,\delta_k)
\end{align*}
where $c$ is a constant, the first equality comes from (\ref{eq:e-tilde}), and the third equality comes from the hypothesis assumption. This proves the claim.
\end{proof}

Note that, given the state estimate used in the structure of the controller, we proved in Lemmas~\ref{lemma1} and \ref{lemma3} that the estimation error and the estimation mismatch satisfy linear recursive equations\index{equation! linear recursive}, which were then used in the derivations of Lemmas~\ref{lemma2} and \ref{lemma4}, respectively. These linear recursive equations, as we will show, lead to a tractable analysis for the characterization of the optimal triggering policy.

In the next lemma, we introduce a loss function that is equivalent to the original loss function, in the sense that optimizing it is equivalent to optimizing the original loss function.

\begin{lemma}\label{lemma5}
Given the adopted certainty-equivalent controller, the following loss function is equivalent to the original loss function $\Phi(\pi)$:
\begin{align}\label{eq:phiprime}
	\Psi(\pi) = \EXP \Big[ \textstyle \sum_{k=0}^{N} \theta \delta_k +  \big \| e_k \big\|^2_{\Gamma_k} \Big]
\end{align}
where $\theta = \ell (1-\lambda)/\lambda$ is a weighting coefficient and $\Gamma_k = A^T S_{k+1} B (B^T S_{k+1} B + R)^{-1} B^T S_{k+1} A$ for $k \in \mathcal{K}$ is a weighting matrix.
\end{lemma}

\begin{proof}
Given (\ref{c1:eq:sys}) and (\ref{eq:riccati}), we can derive the following identities:
\begin{align*}
\begin{split}
	&x_{k+1}^T S_{k+1} x_{k+1} = (A x_k + B u_k + w_k)^T S_{k+1} (A x_k + B u_k + w_k)
\end{split}\\[2.75\jot]
\begin{split}
	&x_k^T S_k x_k = x_k^T \big(Q + A^T S_{k+1} A - L_k^T (B^T S_{k+1} B + R) L_k\big) x_k
\end{split}\\[2.75\jot]
\begin{split}
	&x_{N+1}^T S_{N+1} x_{N+1} - x_0^T S_0 x_0 =  \textstyle \sum_{k=0}^{N} \big( x_{k+1}^T S_{k+1} x_{k+1} - x_k^T S_k x_k \big).
\end{split}
\end{align*}
Using the above identities together with $u_k = -L_k \hat{x}_k$, and applying few algebraic operations, we can obtain $\Psi(\pi)$ as in (\ref{eq:phiprime}), and can see that it is equivalent to $\Phi(\pi)$.
\end{proof}

Based on the loss function $\Psi(\pi)$, we can form the value function\index{value function} $V_k(\mathcal{I}^e_k)$ as
\begin{align}
	V_k(\mathcal{I}^e_k) = \min_{\pi \in \mathcal{P}}\EXP \Big[ \textstyle \sum_{t=k}^{N} \theta \delta_t +  \big \| e_{t+1} \big\|^2_{\Gamma_{t+1}} \big| \mathcal{I}^e_k \Big] \label{eq:Ve-def}
\end{align}
for $k \in \mathcal{K}$, where $\Gamma_{N+1} = 0$ by convention. Given $V_k(\mathcal{I}^e_k)$, we can formally define the value of information\index{value of information}, a metric that measures the difference between the benefit and the cost of the system at each time.

\begin{definition}\label{def:voi}
The value of information at time $k$ is the variation in the value function $V_k(\mathcal{I}^e_k)$ with respect to the information $x_{k-\zeta_k}$ that can be communicated to the controller, i.e.,
\begin{align}\label{eq:voi-def}
\voi_k = V_k(\mathcal{I}^e_k)|_{\delta_k = 0} - V_k(\mathcal{I}^e_k)|_{\delta_k = 1}
\end{align}
where $V_k(\mathcal{I}^e_k)|_{\delta_k}$ denotes $V_k(\mathcal{I}^e_k)$ when $\delta_k$ is enforced.
\end{definition}

We are now in a position to present our main results. We first have the following theorem in which we obtain the optimal triggering policy under the main information structure\index{information structure}, consisting of $\mathcal{I}^e_k$ and $\mathcal{I}^c_k$.

\begin{theorem}\label{thm:1}
The optimal triggering policy $\pi^{\star}$ under the main information structure is a symmetric threshold policy given by $\delta_{k}^{\star} = \mathds{1}_{\voi_k \geq 0}$, where $\voi_k$ is the value of information expressed as
\begin{align}\label{eq:voi-original}
	\voi_k &= \big\| \tilde{e}_k \big\|^2_{A^T \Gamma_{k+1} A} - \theta + \varrho_k(\tilde{e}_k)
\end{align}
and $\varrho_k(\tilde{e}_k) = \EXP[V_{k+1}(\mathcal{I}^e_{k+1})|\mathcal{I}^e_k, \delta_k = 0] - \EXP[V_{k+1}(\mathcal{I}^e_{k+1})|\mathcal{I}^e_k, \delta_k = 1]$.
\end{theorem}

\begin{proof}
From the additivity of $V_k(\mathcal{I}^e_k)$, we obtain
\begin{align*}
	V_k(\mathcal{I}^e_k) = \!\min_{\Prob(\delta_k|\mathcal{I}_k^e)} \! \EXP\Big[\theta \delta_k + e_{k+1}^T \Gamma_{k+1} e_{k+1} + V_{k+1}(\mathcal{I}^e_{k+1}) \big|\mathcal{I}^e_k\Big]
\end{align*}
with initial condition $V_{N+1}(\mathcal{I}^e_{N+1}) = 0$. We prove that $V_k(\mathcal{I}^e_k)$ is a symmetric function of $\tilde{e}_k$. We assume that the claim holds at time $k+1$, and shall prove that it also holds at time $k$. By Lemma~\ref{lemma2}, we find
\begin{align*}
	\EXP \Big[e_{k+1}^T &\Gamma_{k+1} e_{k+1} \big| \mathcal{I}^e_k \Big]= \underset{\delta_k}{\EXP} \Big[ (1- \delta_k) \tilde{e}_k^T A^T \Gamma_{k+1} A \tilde{e}_k + \tr(\Gamma_{k+1} \Sigma_k) \big| \mathcal{I}^e_k \Big]
\end{align*}
where $\Sigma_k = \sum_{t=0}^{\zeta_k} A^t W {A^t}^T$. We can show that
\begin{align}\label{c1:eq:cost-to-go2-perfect}
	V_{k}(\mathcal{I}^e_k) &= \min_{\delta_k} \Big\{\theta \delta_k + (1-\delta_k) \tilde{e}_k^T A^T \Gamma_{k+1} A \tilde{e}_k \nonumber\\
	&\qquad \qquad \qquad + \tr(\Gamma_{k+1} \Sigma_k) + \EXP[V_{k+1}(\mathcal{I}^e_{k+1})|\mathcal{I}^e_k]  \Big\}.
\end{align}
The minimizer in (\ref{c1:eq:cost-to-go2-perfect}) is obtained as $\delta_{k}^{\star} = \mathds{1}_{\voi_k \geq 0}$, where
\begin{align*}
	\voi_k = \tilde{e}_k^T A^T \Gamma_{k+1} A \tilde{e}_k - \theta + \varrho_k
\end{align*}
and $\varrho_k = \EXP[V_{k+1}(\mathcal{I}^e_{k+1})|\mathcal{I}^e_k, \delta_k = 0] - \EXP[V_{k+1}(\mathcal{I}^e_{k+1})|\mathcal{I}^e_k, \delta_k = 1]$. Besides, by Lemma~\ref{lemma4}, we know that $\EXP[V_{k+1}(\mathcal{I}^e_{k+1}) |\mathcal{I}^e_k, \delta_k]$ is a symmetric function of $\tilde{e}_k$. Hence, we conclude that $V_k(\mathcal{I}^e_k)$ is a symmetric function of $\tilde{e}_k$. This completes the proof.
\end{proof}

Theorem~\ref{thm:1} states that under the main information structure the value of information $\voi_k$, as characterized in (\ref{eq:voi-original}), depends on the sample path of the process, and is a symmetric function of the estimation mismatch $\tilde{e}_k = \textstyle \sum_{t = \zeta_k + 1}^{\eta_k} A^{t -1} w_{k - t}$. Moreover, it states that the information $x_{k-\zeta_k}$ should be transmitted from the sensor to the controller at time $k$ only if the value of information $\voi_k$ is nonnegative. Note that the value of information $\voi_k$ in this case can be computed with arbitrary accuracy by solving (\ref{c1:eq:cost-to-go2-perfect}) recursively and backward in time.

We show in following corollary that given the triggering policy in Theorem~\ref{thm:1} the minimum mean-square-error state estimate at the controller in fact matches with the state estimate used in the structure of the controller.

\begin{corollary}
Given the triggering policy $\pi^\star$, the conditional mean $\EXP[x_k | \mathcal{I}^c_k]$ satisfies (\ref{eq:est-at-cont}) with $\imath_k = 0$ for all $k \in \mathcal{K}$.
\end{corollary}

\begin{proof}
Note that $\hat{x}_0 = \EXP[x_0 | \mathcal{I}^c_0]$ holds regardless of the adopted triggering policy. Assume that $\imath_t = 0$ holds for all $t < k$. Hence, we have $\hat{x}_t = \EXP[x_t | \mathcal{I}^c_t]$ for all $t \leq k$. We shall prove that the claim also holds at time $t = k$. We can write
\begin{align}
	\Prob(\tilde{e}_k | \mathcal{I}_k^c, \delta_k = 0) \propto \Prob(\delta_k = 0 |\tilde{e}_k)\Prob(\tilde{e}_k | \mathcal{I}_k^c).
\end{align}
By the hypothesis assumption, Lemma~\ref{lemma3}, and Theorem~\ref{thm:1}, we see that $\Prob(\tilde{e}_k | \mathcal{I}_k^c)$  and $\Prob(\delta_k = 0 | \tilde{e}_k)$ are symmetric with respect to $\tilde{e}_k$. Hence, $\Prob(\tilde{e}_k | \mathcal{I}_k^c, \delta_k = 0)$ is also symmetric with respect to $\tilde{e}_k$. This means that $\EXP[\tilde{e}_k | \mathcal{I}_k^c, \delta_k = 0] = 0$. Besides, we can write
\begin{align*}
\EXP\Big[ e_k \big| \mathcal{I}_k^c, \delta_k \Big] &= \EXP \Big[ \EXP\Big[ e_k \big| \mathcal{I}_k^e\Big] \big| \mathcal{I}_k^c, \delta_k \Big] =\EXP\Big[ \tilde{e}_k \big| \mathcal{I}_k^c, \delta_k \Big]
\end{align*}
where in the first equality we used the tower property of conditional expectations and the fact that $\delta_k$ is a function of $\mathcal{I}^e_k$. Hence, $\imath_k = A ( \EXP[ x_k | \mathcal{I}^c_k, \delta_k =0 ] - \EXP[ x_k | \mathcal{I}^c_k]) = A \EXP [ e_k | \mathcal{I}_k^c, \delta_k = 0 ] = A \EXP[ \tilde{e}_k | \mathcal{I}_k^c, \delta_k = 0 ] = 0$. This completes the proof.
\end{proof}

In the rest of this section, we prove that there exists a condition associated with the information structure under which the value of information becomes completely expressible in terms of the age of information. Let the information set of the event trigger be a restricted set that includes only the timestamps of the informative observations, the timestamps of the communicated informative observations, and the previous decisions, i.e., $\mathcal{I}^{r}_k = \{ t-\zeta_t, t-\eta_t, \delta_s, u_s \ | \ 0 \leq t \leq k, 0 \leq s < k \}$, where for clarity we represented the set by $\mathcal{I}^r_k$ instead of $\mathcal{I}^e_k$. A triggering policy $\pi$ is now admissible if $\pi = \{ \Prob(\delta_k | \mathcal{I}^r_k) \}_{k=0}^{N}$, where $ \Prob(\delta_k | \mathcal{I}^r_k)$ is a Borel measurable transition kernel. We now have the following theorem in which we obtain the optimal triggering policy under the restricted information structure\index{information structure}, consisting of $\mathcal{I}^r_k$ and $\mathcal{I}^c_k$.

\begin{theorem}\label{thm2}
The optimal triggering policy $\pi^{\star}$ under the restricted information structure is a threshold policy given by $\delta_{k}^{\star} = \mathds{1}_{\voi_k \geq 0}$, where $\voi_k$ is the value of information expressed as
\begin{align}\label{eq:voi-age-version}
	\voi_k = \tr \big( \Gamma_{k+1} \textstyle \sum_{t=\zeta_{k}+1}^{\eta_k} A^t W {A^t}^T \big) - \theta + \varrho_k(\zeta_k,\eta_k)
\end{align}
and $\varrho_k(\zeta_k,\eta_k) = \EXP[V_{k+1}(\mathcal{I}^r_{k+1})|\mathcal{I}^r_k, \delta_k = 0] - \EXP[V_{k+1}(\mathcal{I}^r_{k+1})|\mathcal{I}^r_k, \delta_k = 1]$.
\end{theorem}

\begin{proof}
From the additivity of $V_k(\mathcal{I}^r_k)$, we obtain
\begin{align*}
	V_k(\mathcal{I}^r_k) = \!\min_{\Prob(\delta_k|\mathcal{I}_k^r)} \! \EXP\Big[\theta \delta_k + e_{k+1}^T \Gamma_{k+1} e_{k+1} + V_{k+1}(\mathcal{I}^r_{k+1}) \big|\mathcal{I}^r_k\Big]
\end{align*}
with initial condition $V_{N+1}(\mathcal{I}^r_{N+1}) = 0$. We prove that $V_k(\mathcal{I}^r_k)$ is a function of $\zeta_k$ and $\eta_k$. We assume that the claim holds at time $k+1$, and shall prove that it also holds at time $k$. Using Lemma~\ref{lemma1}, and applying few operations, we can write
\begin{align*}
	\EXP[e_{k+1} | \mathcal{I}^r_k] &= 0\\[1.75\jot]
	\Cov[e_{k+1} | \mathcal{I}^r_k] &= \delta_k \textstyle \sum_{t=0}^{\zeta_k} A^t W {A^t}^T + (1-\delta_k) \textstyle \sum_{t=0}^{\eta_k} A^t W {A^t}^T.
\end{align*}
This implies that
\begin{align*}
	\EXP \Big[e_{k+1}^T &\Gamma_{k+1} e_{k+1} \big| \mathcal{I}^r_k \Big] = \underset{\delta_k}{\EXP} \Big[ \delta_k \tr\big( \Gamma_{k+1} \textstyle \sum_{t=0}^{\zeta_k} A^t W {A^t}^T \big) \\[1\jot]
	&\qquad \qquad \qquad \qquad + (1-\delta_k) \tr \big( \Gamma_{k+1} \textstyle \sum_{t=0}^{\eta_k} A^t W {A^t}^T \big) \big| \mathcal{I}^r_k \Big].
\end{align*}
Hence, we can show that
\begin{align}
	V_{k}(\mathcal{I}^r_k) &= \min_{\delta_k} \Big\{\theta \delta_k + \delta_k \tr\big( \Gamma_{k+1} \textstyle \sum_{t=0}^{\zeta_k} A^t W {A^t}^T \big) \nonumber\\[0\jot]
	&\qquad + (1-\delta_k) \tr \big( \Gamma_{k+1} \textstyle \sum_{t=0}^{\eta_k} A^t W {A^t}^T \big) + \EXP[V_{k+1}(\mathcal{I}^r_{k+1})|\mathcal{I}^r_k]  \Big\}.\label{eq:optimaltity-eq-restricted}
\end{align}
The minimizer in (\ref{eq:optimaltity-eq-restricted}) is obtained as $\delta_{k}^{\star} = \mathds{1}_{\voi_k \geq 0}$, where
\begin{align*}
	\voi_k &= \tr \big( \Gamma_{k+1} \textstyle \sum_{t=0}^{\eta_k} A^t W {A^t}^T \big) - \tr\big( \Gamma_{k+1} \textstyle \sum_{t=0}^{\zeta_k} A^t W {A^t}^T \big) - \theta + \varrho_k
\end{align*}
and $\varrho_k = \EXP[V_{k+1}(\mathcal{I}^r_{k+1})|\mathcal{I}^r_k, \delta_k = 0] - \EXP[V_{k+1}(\mathcal{I}^r_{k+1})|\mathcal{I}^r_k, \delta_k = 1]$. Note that $\zeta_{k+1}$ and $\eta_{k+1}$ are functions of $\tau_{k+1}$, $\zeta_k$, and $\eta_k$. Since $\tau_{k+1}$ is an independent random variable with known distribution, we conclude that $V_k(\mathcal{I}^r_k)$ is a function of $\zeta_k$ and $\eta_k$. This completes the proof.
\end{proof}

Theorem~\ref{thm2} states that under the restricted information structure the value of information $\voi_k$, as characterized in (\ref{eq:voi-age-version}), becomes independent of the sample path of the process, and can be expressed in terms of the age of information at the event trigger $\zeta_k$ and that at the controller $\eta_k$. Note that the value of information $\voi_k$ in this case can be computed with arbitrary accuracy by solving (\ref{eq:optimaltity-eq-restricted}) recursively and backward in time. Finally, we remark that using the triggering policy in Theorem~\ref{thm2} instead of the triggering policy in Theorem~\ref{thm:1} leads to a degradation in the performance of the system because $\mathcal{I}^r_k$ contains less information than $\mathcal{I}^e_k$.

As before, given the triggering policy in Theorem~\ref{thm2} the minimum mean-square-error state estimate at the controller matches with the state estimate used in the structure of the controller.

\section{Numerical Example}\label{sec4}
Consider a scalar stochastic process defined by the model in (\ref{c1:eq:sys}) with state coefficient $A=1.1$, input coefficient $B=1$, noise variance $W = 1$, and mean and variance of the initial condition $m_0 = 0$ and $M_0 = 1$. Let the time horizon be $N = 200$, and the weighting coefficients be $\theta = 10$, $Q = 1$, and $R = 0.1$. Moreover, let the output of the process be given by the model in (\ref{outputmodelused}) with delay $d = 5$ and probability $p_k = 0.2$. For this system, we implemented the triggering policy characterized by Theorem~\ref{thm:1}. For a test run of the system, Figure \ref{fig:errors} shows the estimation error norm trajectories at the controller and at the event trigger, and Figure \ref{fig:AoIs} shows the age of information trajectories at the controller and at the event trigger. Furthermore, Figure \ref{fig:VoI} depicts the value of information and transmission event trajectories.

\begin{figure}[t!]
\center
\vspace{-2.5mm}
  \includegraphics[width= 0.95\linewidth]{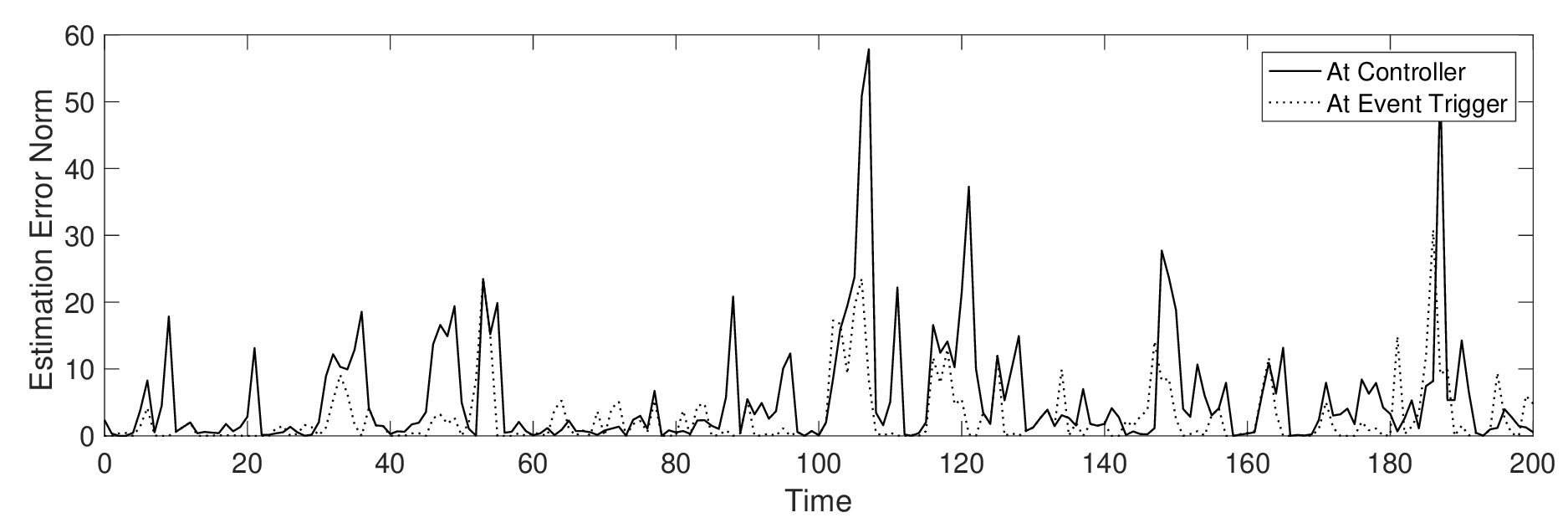}
  \caption{The estimation error norm trajectories at the controller and at the event trigger.}
  \label{fig:errors}
\end{figure}

\begin{figure}[t!]	
\center
\vspace{-2.5mm}
  \includegraphics[width= 0.95\linewidth]{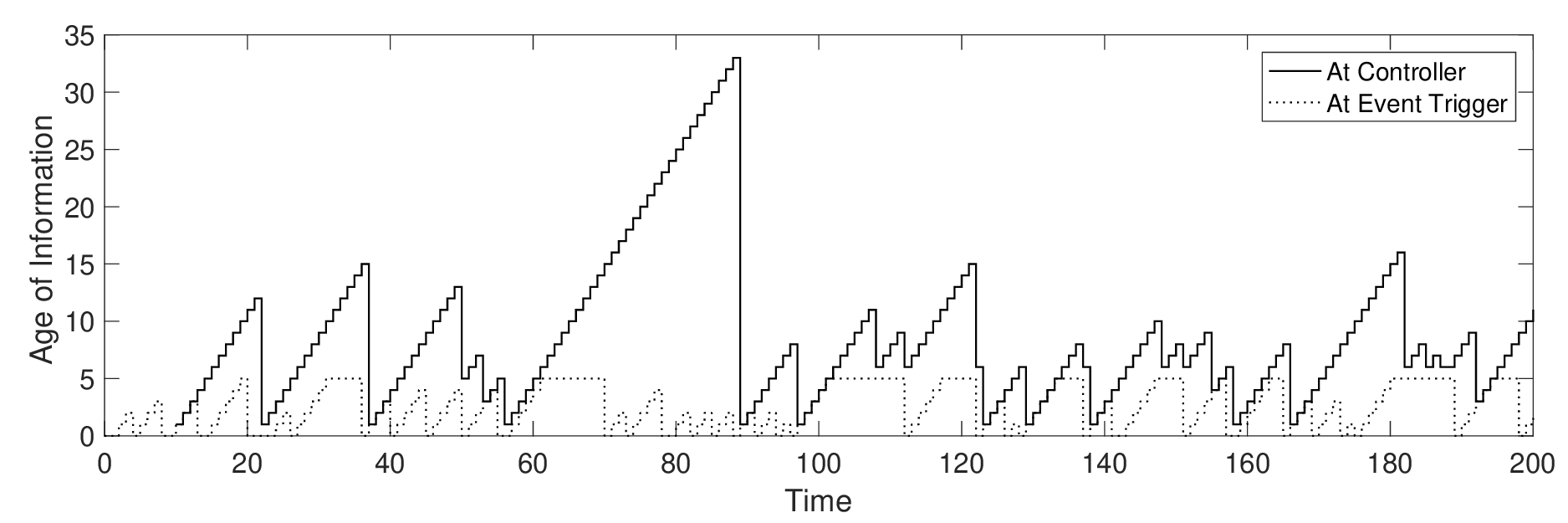}
  \caption{The age of information trajectories at the controller and at the event trigger.}
  \label{fig:AoIs}
\end{figure}

\begin{figure}[t!]	
\center
\vspace{-2.5mm}
  \includegraphics[width= 0.95\linewidth]{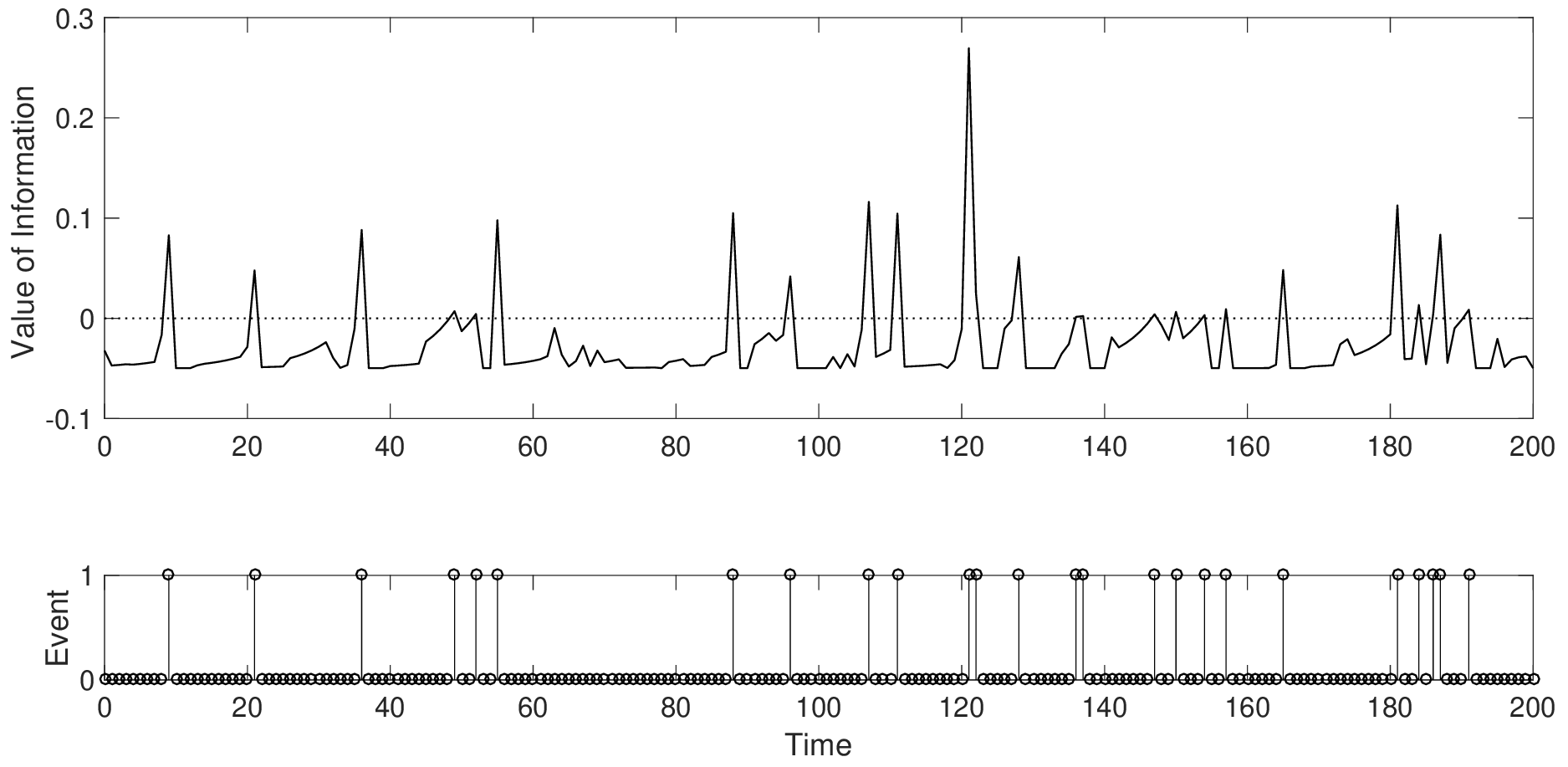}
  \caption{The value of information (scaled by $1/N$) and transmission event trajectories.}
  \label{fig:VoI}
\end{figure}

We recall that the optimal triggering policy in Theorem~\ref{thm:1} permits a transmission only when the value of information is nonnegative. The value of information itself is a function of the estimation mismatch. This implies that the value of information does not necessarily increase with the age of information at the controller or that at the event trigger. This key point can be observed in our numerical example. Consider the time interval $k\in [56,87]$. In this interval, the age of information at the controller increases continually. However, the value of information remains small, and no transmission occurs in this interval. This situation continues until time $k=88$ at which the value of information becomes nonnegative, due to the dramatic increment of the estimation mismatch. As a result, at this time a transmission occurs.

\section{Conclusion}\label{sec5}
In this chapter, we studied the value of information as a more comprehensive instrument than the age of information for optimally shaping the information flow in a networked control system. Two information structures were considered: the main one and a restricted one. We proved that under the main information structure the value of information is a symmetric function of the estimation mismatch. Moreover, we proved that under the restricted information structure the value of information becomes completely expressible in terms of the age of information. Accordingly, we characterized two optimal triggering policies of the form $\delta_k^\star = \mathds{1}_{\voi_k \geq 0}$. One policy is dependent on the sample path of the process and achieves the best performance, while the other one is independent of the sample of the process and achieves a lower performance.

\bibliography{../../../../mybib}
\bibliographystyle{ieeetr}

\end{document}